\documentclass[a4paper,UKenglish]{article}

\usepackage{microtype}%if unwanted, comment out or use option "draft"
\usepackage{tikz}
\usepackage{amsmath}
\usepackage{amsthm}
\usepackage{amssymb}
\usepackage{xcolor}
\usepackage{caption}
\usepackage{listings}
\usepackage{hyperref}

\newcommand*{\thead}[1]{\multicolumn{1}{c}{\bfseries #1}}
\newcommand{\specialcell}[2][l]{%
  \begin{tabular}[#1]{@{}l@{}}#2\end{tabular}}

\newcommand{\dist}[2][]{\mathbf{dist}_{#1}(#2)}
\newcommand{\footremember}[2]{%
    \footnote{#2}
    \newcounter{#1}
    \setcounter{#1}{\value{footnote}}%
}

\title{On the complexity of the (approximate) nearest colored node problem}

\author{Maximilian Probst \footremember{cph}{BARC, University of Copenhagen, Universitetsparken 5, Copenhagen 2100, Denmark, The author is supported by Basic Algorithms Research Copenhagen (BARC), supported by Thorup's Investigator Grant from the Villum Foundation under Grant No. 16582.}}

\newtheorem{lemma}{Lemma}
\newtheorem{theorem}{Theorem}
\newtheorem{definition}{Definition}

\begin{document}

\maketitle

\begin{abstract}
Given a graph $G=(V,E)$ where each vertex is assigned a color from the set $C=\{c_1, c_2, .., c_\sigma\}$. In the (approximate) nearest colored node problem, we want to query, given $v \in V$ and $c \in C$, for the (approximate) distance $\widehat{\mathbf{dist}}(v, c)$ from $v$ to the nearest node of color $c$. For any integer $1 \leq k \leq \log n$, we present a Color Distance Oracle (also often referred to as Vertex-label Distance Oracle) of stretch $4k-5$ using space $O(kn\sigma^{1/k})$ and query time $O(\log{k})$. This improves the query time from $O(k)$ to $O(\log{k})$ over the best known Color Distance Oracle by Chechik \cite{DBLP:journals/corr/abs-1109-3114}.

We then prove a lower bound in the cell probe model showing that even for unweighted undirected paths any static data structure that uses space $S$ requires at least $\Omega\left(\log\frac{\log{\sigma}}{\log(S/n)+\log\log{n}}\right)$ query time to give a distance estimate of stretch $O(\text{polylog}(n))$. This implies for the important case when $\sigma = \Theta(n^{\epsilon})$ for some constant $0 < \epsilon < 1$, that our Color Distance Oracle has asymptotically optimal query time in regard to $k$, and that recent Color Distance Oracles for trees \cite{tsur2016succinct} and planar graphs \cite{DBLP:journals/corr/MozesS15} achieve asymptotically optimal query time in regard to $n$.

We also investigate the setting where the data structure additionally has to support color-reassignments. We present the first Color Distance Oracle that achieves query times matching our lower bound from the static setting for large stretch yielding an exponential improvement over the best known query time \cite{DBLP:journals/corr/abs-1305-3314}. Finally, we give new conditional lower bounds proving the hardness of answering queries if edge insertions and deletion are allowed that strictly improve over recent bounds in time and generality.
\end{abstract}

\section{Introduction}

In the \textit{static} nearest colored node problem, we are given a graph $G=(V,E)$ and a color set $C = \{c_1, c_2, .., c_\sigma\}$ with a function $c : V \xrightarrow{} C$ mapping each vertex to a color in $C$ and we want to compute and store the distance $\mathbf{dist}_G(u, c)$ denoting the distance from each vertex $u \in V$ to the nearest vertex $v$ of color $c \in C$ or more formally $\mathbf{dist}_G(u, c) = \mathbf{min}_{v \in V_c} \mathbf{dist}(u,v)$ where $V_c$ denote the set of $c$-colored vertices. Clearly, if the color set $C$ is of small cardinality, computing distances and storing a distance matrix is an efficient measure but for large $\sigma$, i.e. $\sigma = \Omega(n^{\epsilon})$ for some $\epsilon > 0$, this solution becomes impractical for many important applications. The problem has various applications in navigation, routing and document processing, often in connection to locating resources or facilities, e.g. the nearest gas station, quickly. 

In this article, we present a new data structure that reports for any given $v \in V, c \in C$ a distance estimate $\widehat{\mathbf{dist}}(v,c)$ in $O(\log{k})$ time such that $\mathbf{dist}(v,c) \leq \widehat{\mathbf{dist}}(v,c) \leq (4k-5)\mathbf{dist}(v,c)$ using $O(kn\sigma^{1/k})$ space for every positive integer $k$, improving on $O(k)$ query time of the previous data structure with same space consumption and stretch factor \cite{DBLP:journals/corr/abs-1109-3114}. We refer to the data structure as \textit{Color Distance Oracle}. In other literature, the problem is also studied as the vertex-to-label problem and the data structure is referred to as Vertex-label Distance Oracle \cite{Hermelin2011}\cite{DBLP:journals/corr/abs-1109-3114}. The term Distance Oracle originates from the classic Thorup-Zwick Distance Oracle \cite{Thorup:2005:ADO:1044731.1044732} that reports distance estimates for each pair of vertices. Thorup, Zwick and Roddity \cite{roditty2005deterministic} also generalized Distance Oracles by introducing source-restricted Distance Oracles where only distances in $S \times V$ can be queried for some subset $S \subseteq V$ and that uses space $O(n|S|^{1/k})$ for stretch $2k-1$. It is tempting to approach the nearest colored node problem by using an auxiliary vertex for each color $c \in C$ linking it to all $c$-colored vertices with a zero-weight edge and let the auxiliary vertices define the set $S$. Unfortunately, this might decrease some vertex-to-color distances by creating "portals" through the auxiliary vertices. Instead the underlying sampling techniques can, together with a more advanced analysis stemming from Compact Routing  Schemes \cite{Thorup:2001:CRS:378580.378581}, be used to construct efficient and correct Color Distance Oracles. Color Distance Oracles can also be seen as a generalization of (source-restricted) Distance Oracles as we can choose a color set of size $|S|$ and assign each vertex in $S$ a unique color.

\begin{table}
\scriptsize
\begin{tabular}{lllll} 
 \thead{Graph Family} & \thead{Approximation} & \thead{Space} & \thead{Query Time} & \thead{Ref}  \\ \hline
\specialcell[l]{Unweighted\\undirected path} & $O(\text{polylog}(n))$ & $S$ & $\Omega\left(\log\frac{\log{\sigma}}{\log(S/n)+\log\log{n}}\right)$ & \textbf{New}\\ \hline
 Tree & Exact & $O(n)$ & $O\left(\log\frac{\log{\sigma}}{\log{w}}\right)$ & \cite{tsur2016succinct}  \\ \hline
 Planar graphs & $1+\epsilon$ & $O(n \log n )$  & $O(\log\log n) $ & \cite{DBLP:journals/corr/MozesS15}   \\ \hline
 \specialcell[l]{Planar\\ digraphs} & $1+\epsilon$ & $O(n \log n)$  & \specialcell[l]{$O(\log\log n$\\$\log\log(nN))$} & \cite{DBLP:journals/corr/MozesS15}   \\ \hline
 General graphs  & $4k-5$ & $O(kn\sigma^{1/k})$ & $O(\log{k})$ & \textbf{New}\\ \hline
 General graphs  & $8(1+\epsilon)k$ & $O(kn^{1+1/k}\log{n})$ & $O\left(\log\frac{\log{\sigma}}{\log{w}}\right)$ & \textbf{New}
\end{tabular}
\caption{Best upper and lower bounds for static Color Distance Oracles. $N$ denotes the weight of the heaviest edge in the graph. The $w$ in the second column refers to the word size and we assume in this article that $w = \Theta(\log n)$. $N$ refers to the heaviest edge weight in the graph.}
\label{fig:staticlandubounds}
\end{table}

We first present a Color Distance Oracle of space $O(kn\sigma^{1/k})$ that reports in $O(\log{k})$ time distance estimate of stretch at most $(4k-5)$. Our Color Distance Oracle matches space and stretch of the best data structure by Chechik \cite{DBLP:journals/corr/abs-1109-3114} and improves the query time from $O(k)$ to $O(\log{k})$. This is in fact achieved by combining Chechiks result with a well-known technique by Wulff-Nilsen \cite{DBLP:journals/corr/abs-1202-2336} for general distance oracles. Our contribution is to simplify the technique of Wulff-Nilsen by observing that a Range Minimum Query (RMQ) data structure can be used to replace his tailor-made data structure even more efficiently resulting in a concise and simple algorithm; and to generalize the proof technique of Wulff-Nilsen. Recently, Chechik has also shown that classic Distance Oracles can be implemented with constant query time \cite{chechik2015approximate}\cite{DBLP:journals/corr/abs-1305-3314} and it is natural to ask whether this improvement carries over to Color Distance Oracles. Our new lower bound rules out such an improvement and shows that our Color Distance Oracle has essentially tight query time when $\sigma = \Theta(n^{\epsilon})$ for constant $0 < \epsilon < 1$, as the lower bound then simplifies to $\Omega(\log(\epsilon k))$ for $S=n^{1+1/k}$ for all values of $k = O(\text{polylog}(n))$. Our result extends to prove asymptotic optimality in query time in regard to $n$ for the best known Color Distance Oracles for trees \cite{gawrychowski2016nearest}\cite{tsur2016succinct} and planar graphs \cite{DBLP:journals/corr/MozesS15} even for data structures with higher stretch $k = O(\text{polylog}(n))$. This lower bound, that is our main contribution, is thus a significant step in understanding Color Distance Oracles and their limitations. An overview over the best upper bounds for different graph families and our lower bound is given in table \ref{fig:staticlandubounds}.

We also present a new Color Distance Oracle for the setting where the data structure needs to handle color-reassignments, i.e. updates in which a vertex $v \in V$ is assigned a new color $c \in C$ such that afterwards $c(v) = c$. Our Color Distance Oracle is conceptually simple, building on some recent results in Ramsey theory\cite{DBLP:journals/corr/AbrahamCEFN17} and can be constructed deterministically. It strictly improves on query and update time for any approximation factor $k = \Omega(\frac{\log n}{\log\log n})$ and dominates existing data structures in query time for $k = \Omega(\log\log n)$. We are also able to show an elegant trade-off between query and update time that was unknown before. For $k = \Omega(\frac{\log n}{\log\log n})$ our data structure requires $\tilde{O}(n)$ space and updates only take polylogarithmic time. Therefore our static lower bound extends to this setting as we can start with an uncolored graph and color vertices in $n$ updates. This implies that our query time is tight with regard to $n$. Achieving query time $O({\log\log n})$ is rather surprising given that queries for exact distances take $\Omega(\frac{\log n}{\log\log n})$ time even on unweighted balanced trees and that for the static setting approximation doesn't admit any query time improvements. An overview over upper bounds and lower bounds for the color-reassignment setting is given in table  \ref{fig:dynlandubounds}.

\begin{table}[ht]
\scriptsize
\begin{tabular}{lllll} 
 \thead{Graph Family} & \thead{\specialcell{Approx-\\imation}} & \thead{Update Time}& \thead{Query Time} & \thead{Ref} \\ \hline
   \specialcell{Unweighted \\undirected path} & $O(\text{polylog}(n))$ & $O(\text{polylog}(n))$ & $\Omega\left(\log\frac{\log{\sigma}}{\log(S/n)+\log\log{n}}\right)$ & \textbf{New} \\ \hline
 \specialcell{Unweighted \\tree} & exact & $O(\text{polylog}(n))$ & $\Omega\left(\frac{\log(n)}{\log\log(n)}\right)$ & \cite{gawrychowski2016nearest} \\ \hline
 Tree & exact  & $O(\log n)$ & $O\left(\log n\right)$ & \cite{gawrychowski2016nearest}  \\
  & exact & $O(\log^{1+\epsilon} n )$ & $O\left(\frac{\log(n)}{\log\log(n)}\right)$ & \cite{gawrychowski2016nearest} \\ 
 \hline
 \specialcell{Planar\\ graphs} & $1+\epsilon$ & \specialcell[l]{$O(\epsilon^{-1} \log(n)$\\$ \log\log(n))$} & \specialcell[l]{$O(\epsilon^{-1} \log(n)\log(nN)$\\$\log\log(n))$} & \cite{laish2017efficient}    \\ 
  & $1+\epsilon$ & \specialcell[l]{$O\left(\epsilon^{-1} \frac{\log^2(\epsilon^{-1}n)}{\log\log(n)}\right)$} & \specialcell[l]{$O(\epsilon^{-1} \log^{1.51}(\epsilon^{-1}n)$)} & \cite{laish2017efficient}  \\ \hline
 \specialcell{Planar\\ digraphs} & $1+\epsilon$ & \specialcell[l]{$O(\epsilon^{-1} \log(n)$\\$\log\log(nN))$} & $O(\epsilon^{-1}\log^3(n)\log(nN))$ & \cite{laish2017efficient}  \\ \hline
 \specialcell{General\\ graphs} & $<3$ & $\Omega(n^{1-\epsilon})$&$\Omega(n^{2-\epsilon})$ & \cite{DBLP:journals/corr/HenzingerKNS15}\\
 &$4k-5$ & $O(kn^{1/k}\log^{1-1/k} n$ & $O(k)$ & \cite{DBLP:journals/corr/abs-1109-3114} \\ 
  &   & $\log\log n)$ & &  \\
  & $8(1+\epsilon)k$ & $O\left(\epsilon^{-1}kn^{1/k}\log\log n\right)$ & $O\left(\log\log n \right)$ & \textbf{New} \\
  & $8(1+\epsilon)k$ & $O\left(\log\log n\right)$ &  $O\left(\epsilon^{-1}kn^{1/k} \log\log n\right)$ & \textbf{New} \\
\end{tabular}
\caption{Best upper and lower bounds for Color Distance Oracles supporting color-reassignments. Here $N$ denote the heaviest edge weight in the graph. The lower bound from \cite{DBLP:journals/corr/HenzingerKNS15} applies to update or query time for every $\epsilon > 0$. More precisely, any algorithm with update time $\tilde\tilde{o}(n^{1-\epsilon})$ and $\tilde{\tilde{o}}(n^{2-\epsilon})$ would refute the OMv-conjecture.}
 \label{fig:dynlandubounds}
\end{table}

Finally, we prove that the \textit{dynamic} version of the problem, allowing edge insertions and deletions, cannot process updates in time $\tilde{\tilde{o}}(\sigma)$\footnote{We use the $\tilde{\tilde{o}}(f(n))$-notation as introduced by Henzinger et al. \cite{DBLP:journals/corr/HenzingerKNS15} to denote that the running time is in $O(f(n)^{1-\epsilon})$ for some $\epsilon > 0$. For multiple parameters we let $\tilde{\tilde{o}}(n_1 n_2 n_3)$ be equivalent to $O(n_1^{1-\epsilon}n_2n_3 + n_1n_2^{1-\epsilon}n_3 + n_1n_2n_3^{1-\epsilon})$ for some $\epsilon > 0$.} and queries in time $\tilde{\tilde{o}}(n/\sigma)$ even on unweighted path graphs for queries that ask to report given a fixed source $s \in V$ whether there is a vertex of color $c$ in the same component, unless Online Matrix Multiplication(OMv) has a truly subcubic time algorithm. We then show that even update time $\tilde{\tilde{o}}(\sigma)$ and query time $\tilde{\tilde{o}}(n)$ is not possible if we have directed general graph or ask for distance queries of approximation factor $<5/3$. Combined they strictly improve in generality and query time over a recent lower bound by Gawrychowski et al. \cite{gawrychowski2016nearest} showing that for weighted trees, query and update time $\tilde{\tilde{o}}(\sqrt{n})$ where $\sigma = \Theta(\sqrt{n})$ would imply a truly subcubic solution to tripartite APSP. Our reduction implies an interesting connection to Pagh's problem and the lower bound is in fact obtained by adapting the reduction from Pagh's problem in  \cite{DBLP:journals/corr/HenzingerKNS15}.
\section{Preliminaries}
\label{sec:prelim}
We denote by $\mathbf{dist}(u,v)$ for $u,v \in V$ the shortest-path distance from $u$ to $v$ and by $\mathbf{dist}(u,c)$ with $u \in V, c \in C$ the shortest-path distance between $u$ and the nearest vertex of color $c$. When the context is clear, we often only refer to the nearest colored vertex instead of the nearest vertex of color $c$ and let $c$ denote the color under consideration. We let $v = \textsc{Nearest}(u,c)$ denote the nearest vertex of color $c$ to $u$, i.e. $\mathbf{dist}(u,c) = \mathbf{dist}(u,\textsc{Nearest}(u,c))$. We also make use of the following data structures.

\textbf{Predecessor Search Problem.} In the predecessor search problem, we are given a universe $U=\{0, .., m-1\} = [m]$ and a subset $S \subseteq U$ of size $n = |S|$. Given an element $x \in U$, we ask for the largest element in $S$ that is smaller than $x$ or more formally, we ask for the predecessor of $x$, $\textsc{Pred}(x) = \max\{y \in S| y < x\}$. We let the successor of $x$ be $\textsc{Succ}(x) = \min\{y \in S| y > x\}$. P\v{a}tra\c{s}cu and Thorup present in \cite{DBLP:journals/corr/abs-cs-0603043} a predecessor search data structures that solves queries and updates (insertions and deletions into/from $S$) in $O(\log\log n)$ time and linear space and give a lower bound of $\Omega(\log\log n /\log\log\log n)$ if $m \leq \text{poly}(n)$ and only almost linear-space is given. 

\textbf{Range Minimum Query (RMQ).} A RMQ is a structure augmenting an array $A[1 .. n]$ answering queries of the form $\textsc{Rmq}(i,j) = \min_{k \in [i,j]} A[k]$ by returning the index of field with the minimal value, for any $1 \leq i \leq j \leq n$. RMQ can be solved with $O(n)$ preprocessing time, taking $O(n)$ space and $O(1)$ query time \cite{Bender:2000:LPR:646388.690192}\cite{Fischer2006}.

\textbf{Least Common Ancestor (LCA).} The LCA problem is the problem of finding the least common ancestor in $T$, which we denote $\textsc{Lca}(x,y)$ $T$, of any two nodes $x,y \in \mathcal{V}(T)$. The LCA problem can be reduced to RMQ and can therefore be implemented within the same bounds. 

\textbf{Hash table.} Given a universe $U = [m]$ and a (dynamic) subset $S \in U$, with $n = |S|$, we can query for any $x \in U$ whether $x$ is in $S$. In \cite{dietzfelbinger1994dynamic}, a data structure is presented that can run deterministic queries in constant time and updates of the set $S$, i.e. insertions and deletions, in constant amortized time.

\section{Static Color Distance Oracle}
\label{subsec:fasterquerydo}

 We construct our Color Distance Oracle as in \cite{DBLP:journals/corr/abs-1109-3114} with the classic techniques by Thorup and Zwick \cite{Thorup:2005:ADO:1044731.1044732}: For a given positive integer $k$, we construct the vertex sets $V = A_0 \supseteq A_1 \supseteq A_2 \supseteq \cdots \supseteq A_{k-1}$, where $A_i$ is obtained by sampling each vertex in $A_{i-1}$ with probability $\sigma^{-1/k}$, for $1 \leq i \leq k-1$, and define the set $A_k = \emptyset$. For each vertex $v$ in $V$, we store for each set $A_i$ with $0 \leq i \leq k-1$ the closest neighbour in $A_i$ denoted by $p_i(v)$, where we break ties arbitrarily. For every $v \in V$, we define $\Delta_i(v)= \dist{v, p_{i+1}(v)} - \dist{v, p_{i}}$, for $0 \leq i < k-1$. We then store all such distances in a consecutive array $P_v[0 .. k-2]$ with $P_v[i] = \Delta_i(v)$ for all $i$ and augment $P_v$ by a RMQ-structure that returns the maximum value in a subarray. We denote a query on the RMQ structure over $P_v$ in the range $a$ to $b$ by $\textsc{Rmq}_v(a, b)$ for any $0 \leq a \leq b < k-1$. We define a \textit{bunch} $B(v)$ for every vertex $v$ in $V$ as follows
\[
  B(v) = \bigcup\limits_{i=0}^{k-1} \{u \in A_i  \setminus A_{i+1} | \dist{v,u} < \dist{v, p_{i+1}(v)}\}
\]
and construct for every color $c \in C$ a bunch $B(c)= \bigcup_{v \in V_c} B(v)$. With each $B(c)$, we store in a hash table the vertices $v \in B(c)$ and associate with their key the distance $\dist{v, c}$. This completes our construction. The following lemma bounds space and construction time, but we defer the proof to the appendix since it only differs by bounding the space of the RMQ data structures from the proof in \cite{DBLP:journals/corr/abs-1109-3114}.

\begin{lemma}
 We use at most space $O(kn\sigma^{1/k})$ to represent the Color Distance Oracle and construction time $O(m\sigma)$.
\label{lem:sizebunch}
\end{lemma}

Note that the construction cost can be slightly improved for $\sigma > n^{k/(2k-1)}$ using the construction of Hermelin et al. \cite{Hermelin2011} but our query time improvement doesn't carry over to their construction. We give the following query algorithm for color $c \in C$ and vertex $v \in V$:
\begin{lstlisting}[caption={Query(v,c)},label=list:8-6,captionpos=t,abovecaptionskip=-\medskipamount, mathescape=true]
${lower\_bound} \gets 0$
${upper\_bound} \gets k-1$
While $upper\_bound \neq lower\_bound$ 
Do
   ${i} \gets \lceil({lower\_bound} + {upper\_bound})/2 \rceil$
   // Compute index $j$ s. t. $\Delta_j(v) = \max_{a \in \{{lower\_bound}, .., i - 1\}} \Delta_a(v)$
   ${j} \gets \textsc{Rmq}_v({lower\_bound}, i - 1)$
   If $p_j(v) \not\in B(c)$
   Then
      ${lower\_bound} \gets {i}$
   Else 
      ${upper\_bound} \gets {j}$
   End
End
Return $p_{{lower\_bound}}(v)$
\end{lstlisting}

We claim that the query procedure returns a colored vertex $w_{c}$ for $w_c = p_{{lower\_bound}}(v)$ whose distance to $v$ is $\mathbf{dist}(v, c) \leq \mathbf{dist}(v, w_{c}) \leq (4k-5)\mathbf{dist}(v, c)$. For the rest of the section, we let $w_{best} = \textsc{Nearest}(v,c)$.

As in \cite{DBLP:journals/corr/abs-1202-2336}, we let $\mathcal{I}$ be the sequence $0,.., k-1$. We call an index $j \in \mathcal{I}$, \textit{$(v,c)$-terminal} if $p_j(v) \in B(c)$. We say that a subsequence $i_1,.., i_2$ of $\mathcal{I}$ is \textit{$(v,c)$-feasible} if
(1) $\mathbf{dist}(v, p_{i_1}(v)) \leq 2i_1 \mathbf{dist}(v,c)$, and
(2) $i_2$ is \textit{$(v,c)$-terminal}. Using these definitions we are ready to prove our claim in the ensuing two lemmas.

\begin{lemma}
\label{lemma:lowerbound}
  Let $i_1, .., i_2 \subseteq \mathcal{I}$, with $|\mathcal{I}|>1$, be \textit{$(v,c)$-feasible} and let $i = \lceil({i_1} + {i_2})/2 \rceil$. Let $\mathcal{I'}$ be the sequence $i_1, .., i-1$. Let $j$ be the index in $\mathcal{I}'$, that maximizes $\Delta_j(v)$. Then if $j \not\in B(c)$ the subsequence $i, .., i_2$ is \textit{$(v,c)$-feasible}. Otherwise, the subsequence $i_1, .., j$ is \textit{$(v,c)$-feasible}. The obtained subsequence is of size at most $\frac{2}{3}\mathcal{I}$.
\end{lemma}
\begin{proof}
If $p_j(v) \in B(c)$, then $j$ is \textit{$(v,c)$-terminal}. Hence $i_1, .. j$ is  \textit{$(v,c)$-feasible}. As $j \leq  \lceil({i_1} + {i_2})/2 \rceil - 1 < ({i_1} + {i_2})/2$ and $|\mathcal{I}|>1$ the subsequence is of size at most $\frac{1}{2}\mathcal{I}$.

Now, consider the case where $p_j(v) \not\in B(c)$. Then \[
\mathbf{dist}(w_{best}, p_{j+1}(w_{best})) < \mathbf{dist}(w_{best}, p_{j}(v)).
\] 
We can now employ the analysis from \cite{Thorup:2001:CRS:378580.378581}:
\begin{equation}
\begin{split}
\mathbf{dist}(v, p_{j+1}(v)) &\leq \mathbf{dist}(v, p_{j+1}(w_{best}))
\\&\leq \mathbf{dist}(w_{best},v) + \mathbf{dist}(w_{best}, p_{j+1}(w_{best})) \\&< \mathbf{dist}(w_{best},v) + \mathbf{dist}(w_{best}, p_{j}(v))
\\&\leq 2\mathbf{dist}(w_{best},v) + \mathbf{dist}(v, p_{j}(v))
\end{split}
\end{equation}
Therefore $\Delta_j(v) = \mathbf{dist}(v, p_{j+1}(v)) - \mathbf{dist}(v, p_{j}(v)) \leq  2\mathbf{dist}(w_{best},v)$. Since $i_1,..,i_2$ is \textit{$(v, c)$-feasible}, we have $\mathbf{dist}(v, p_{i_1}(v)) \leq 2i_1 \mathbf{dist}(v, c)$. By choice of $j$, 
\begin{equation}
\begin{split}
\mathbf{dist}(v, p_{i}(v))
&= 2i_1\mathbf{dist}(v, w_{best}) + \sum_{j' \in \mathcal{I}'} \Delta_{j'}(v)\\
&\leq 2i_1\mathbf{dist}(v,w_{best}) + |\mathcal{I}'|\max_{j' \in \mathcal{I}'} \Delta_{j'}(v)\\
&= 2i_1\mathbf{dist}(v, w_{best}) + (i - i_1) \Delta_{j}(v) \\
&= 2i \mathbf{dist}(v, w_{best})
\end{split}
\end{equation}
As $i_2$ is \textit{$(v, c)$-terminal}, we therefore get that $i, .., i_2$ is \textit{$(v, c)$-feasible}. It is now easy to see that by choice of $i$ and as $|\mathcal{I}| > 1$ the derived sequence is smaller $\frac{2}{3}|\mathcal{I}|$.
\end{proof}

\begin{lemma}
The algorithm given in procedure $\textsc{Query}(v,c)$ reports a distance estimate with stretch at most $(4k-5)$ in time $O(\log{k})$.
\end{lemma}
\begin{proof}
  By lemma \ref{lemma:lowerbound} the number of potential indices reduces by factor $\frac{2}{3}$ by every iteration of the loop. Therefore, we have at most $log_{\frac{3}{2}}{k}$ iterations. As querying the RMQ data structure takes constant time the overall running time is $O(\log{k})$.

 To show stretch of at most $4k-5$, we observe that the final sequence has only a single index $j \leq k-1$, and as the sequence is still \textit{$(v, c)$-feasible}, we get by property (1) that $\mathbf{dist}(v, p_j(v)) \leq 2j\mathbf{dist}(v, w_{best}) \leq 2(k-1)\mathbf{dist}(v, w_{best})$. By property (2), we get that $p_j(v)\in B(c)$ and we can bound 
\begin{equation}
\begin{split}
\mathbf{dist}(p_j(v), c) &= \mathbf{dist}(p_j(v), w_{best}) 
\leq \mathbf{dist}(v, w_{best}) + \mathbf{dist}(v, p_j(v)) \\ 
&\leq \mathbf{dist}(v,w_{best}) + 2(k-1)\mathbf{dist}(v, w_{best})
\end{split}
\end{equation}
giving an overall distance of $\mathbf{dist}(v, p_j(v)) + \mathbf{dist}(p_j(v), c) \leq (4k-3)\mathbf{dist}(v, w_{best})$. This can be slightly improved to stretch $4k-5$ by using the technique from \cite{Thorup:2001:CRS:378580.378581}(Lemma A.2) even without changing the overall approach. We only have to adapt property (1) in the definition of \textit{$(v, c)$-feasible} sequences $i_1, .., i_2$ to $\mathbf{dist}(v, p_{i_1}(v)) \leq (2i_1-1) \mathbf{dist}(v, c)$, and adapt lemma \ref{lemma:lowerbound} to directly get the improvement. 
\end{proof}
We point out that the technique presented to query the Color Distance Oracle extends to Compact Routing Schemes as described by Thorup and Zwick \cite{Thorup:2001:CRS:378580.378581} such that paths can be computed in $O(\log{k})$ time. By using a recently devised succinct RMQ structure that is only allowed to query intervals where both indices are multiples of $\log{k}$ by Tsur \cite{tsur2016succinct}, we can adapt our approach to reduce the intervals as described in lemma \ref{lemma:lowerbound} down to size $O(\log{k})$ and then test each remaining index in $O(\log{k})$ overall time. The data structure only requires $O\left(\frac{k \log\log{k}}{\log{k}}\right) = o(k)$ additional bits compared to $O(k\log n)$ bits required for our preceding structure. This is an important improvement for routing schemes as the RMQ structure needs to be appended to each label in order to achieve the query time improvement.

\section{Supporting color-reassignments}
\label{sec:colorreassign}
In this section, we let $(V, \mathbf{dist})$ denote a n-point metric space and let the metric be denoted by $\rho$. We say that $\rho$ is an ultrametric for $V$ if it is a metric that ensures the strong triangle inequality $\rho(x, z) \leq \max\{\rho(x,y),\rho(y,z)\}$ for all $x,y,z \in V$. It is well-known, that a finite ultrametric can be represented by a rooted hierarchically well-separated tree (HST)\footnote{A good introduction to HSTs and metric representations can be found in \cite{Bartal:2003:MRP:780542.780610}.} $T=(V_T,E_T)$ with a value assigned to each vertex in $V_T$ by the function  $\Delta : V_T \rightarrow (0, \infty)$ whose leaf set is $V$ and where $\Delta(v) < \Delta(\textsc{Parent}(v))$ for any $v \in V_T$. Then given $x,y \in V$, $\rho(x,y) = \Delta(\textsc{lca}_T(x,y))$. Thus, working with ultrametrics is very convenient as they allow to reduce problems to problems on trees which are normally well-understood.

In 2005, Mendel and Naor \cite{DBLP:journals/corr/abs-cs-0511084} showed a Las Vegas algorithm that given a metric space $(V,\mathbf{dist})$ finds a ultrametric $\rho$ such that for a subset $U \subseteq V$ of size $|U| \geq |V|^{1-1/k}$ the distortion of the ultrametric would be low, i.e. that for each $u \in U, v \in V, \mathbf{dist}(u,v) \leq \rho(u,v) \leq 128k \mathbf{dist}(u,v)$. They then showed that given this algorithm, a collection of ultrametrics $\mathcal{R} = \{ \phi_1, \phi_2, .., \phi_s\}$ with $E[s]= O(kn^{1/k})$ can be found together with a function $\mathbf{home} : V \rightarrow [1,s]$ such that for each $u,v \in V$ with $i = \mathbf{home}(v)$
\[
\mathbf{dist}(u,v) \leq \rho_{i}(u,v) \leq 128k\mathbf{dist}(u,v)
\]
Recently, Abraham et al. \cite{DBLP:journals/corr/AbrahamCEFN17} gave a deterministic algorithm that improves this result to stretch $8(1+\epsilon)k$ by increasing the size of the collection $\mathcal{R}$ by factor $O(\epsilon^{-1})$. Thus, for fixed $\epsilon$ the space only differs by a constant factor.

For our Color Distance Oracle, we find a collection of ultrametrics $\mathcal{R}$ and build for each ultrametric $\rho_i \in \mathcal{R}$ an HST $T_i$. Additionally, we store with each $v \in V$ a pointer to $T_{\mathbf{home}(v)}$. We observe that taking 
\[ 
\text{min}_{u \in V_c} \Delta(\textsc{Lca}_{T_{\mathbf{home}(v)}}(u,v))
\]
always gives a $8(1+\epsilon)k$ approximation on $\mathbf{dist}(v, c)$ as we take the smallest distance estimate among all estimates to colored vertices and each of them is at least $8(1+\epsilon)k$ approximate in the represented metric $\rho_{\mathbf{home}(v)}$. Let $a$ be the least common ancestor of $v$ and the nearest colored node in $T_{\mathbf{home}(v)}$, i.e. $a = \textsc{Lca}_{T_{\mathbf{home}(v)}}(v, \textsc{Nearest}(v,c))$. Then, there cannot be any vertex $a'$ on the path from $v$ to $a$ in $T_{\mathbf{home}(v)}$ with a colored vertex in its subtree as otherwise $a'$ would be the least common ancestor of $v$ and the colored node and by the property of HSTs that $\Delta(x) < \Delta(\textsc{Parent}(x))$ for all $x \in V$, we would thus derive a contradiction on the minimality of $\Delta(a)$. We conclude that to derive a distance estimate on $\mathbf{dist}(v,c)$, we only need to find the nearest ancestor of $v$ that has a colored node in its subtree. We therefore construct the data structure described in the following lemma \ref{lmm:treecontainlambda} over each HST, whose proof is deferred to the appendix.

\begin{lemma}
\label{lmm:treecontainlambda}
We can maintain a data structure over a tree $T=(V,E)$ with function $c : V \rightarrow C$ as defined before, that given $v \in V, c \in C$ finds the nearest ancestor of $v$ in $T$ that has a $c$-colored vertex in its subtree and that is able to process color-reassignments. Both operations take $O(\log\log{n})$ worst-case time and the data structure requires $O(n)$ space.
\end{lemma}

As every HST requires at most $O(n)$ space and $s = O(\epsilon^{-1}kn^{1/k})$, our data structure requires space $O(\epsilon^{-1}kn^{1+1/k})$. To achieve the data structure with fast query time, we answer a query for $\widehat{\mathbf{dist}}(v,c)$ by querying the data structure described in lemma \ref{lmm:treecontainlambda} on tree $T_{\mathbf{home}(v)}$ returning the nearest colored ancestor $a$ and we return $\Delta(a)$. Thus only a single invocation of the tree structure is required which can be implemented in time $O(\log\log{n})$. For updates, we iterate over all $s$ tree data structures and invoke the color-reassignment on the same parameters. This takes $O(\log\log{n})$ time per tree and as we have $s$ trees, the running time is bound by $O(\epsilon^{-1}\log\log{n}kn^{1/k})$.

For fast update times, we can process a color-reassignment $v \in V, c \in C$ by changing the color in $T_{\mathbf{home}(v)}$ only leaving $v$ without a color in all other trees. If we run our query for $v \in V, c \in C$ on every of the $s$ HSTs, we also know that our distance estimate is a $8(1+\epsilon)k$ approximation as we check every ultrametric and we are sure that for each vertex $u \in V, c(u)=c$ we once queried the HST $T_{\mathbf{home}(u)}$ where $u$ is colored and $\Delta(\textsc{Lca}(v,u))$ is a $8(1+\epsilon)k$ approximation. By standard techniques \cite{DBLP:journals/corr/abs-cs-0511084}, we can also reduce the space consumption to $O(\epsilon^{-1}n^{1+1/k})$ for the data structure with fast updates. The data structure from lemma \ref{lmm:treecontainlambda} can easily be adapted to give a $c$-colored witness $\hat{u}$ such that $\mathbf{dist}(v,\hat{u}) \leq 8(1+\epsilon)k \mathbf{dist}(v,c)$. Finally, if we are only interested in queries, we can for each vertex $v$ of color $c$ in a tree color all its ancestors with $c$ where we allow multiple colors. As HSTs can be balanced, we get an additional factor of $O(\log{n})$ in our data structures. We can then for each color use a succinct nearest marked ancestor structure as presented by Tsur \cite{tsur2016succinct} with $O(\log\frac{\log{\sigma}}{\log{w}})$ query time. 

\section{Lower bound for static Color Distance Oracles}

In this section we prove the following lower bound.
\begin{theorem}
\label{thm:mainresultlbdistancesumry}
Consider an unweighted path $G=(V,E)$ with coloring $c : V \rightarrow C$ and $\sigma \leq O(n^{1-\epsilon})$ for any $\epsilon > 0$. Then, any data structure, using space $S$ on a machine with word size $w = \Theta(\lg{n})$, reporting nearest colored node distance estimates of approximation $k = O(\text{polylog}(n))$ has query time
\[
\Omega\left(\log\frac{\log \sigma}{\log\frac{S}{n}+\log\log n}\right)
\]
The theorem applies to deterministic and randomized queries, admitting a constant error probability.
\end{theorem}

Our proof extends recent results by Gawrychowski et al. \cite{gawrychowski2016nearest} who proved a similar statement for the exact version of the problem. Before we prove our result, we review their lower bound and then show how to extend it. 

The lower bound for the exact version of the problem is based on a reduction from the colored predecessor problem which was used to establish hardness of the predecessor problem \cite{DBLP:journals/corr/abs-cs-0603043,patracscu2007randomization} in the cell-probe model. Belazzougui and Navarro showed in \cite{Belazzougui:2015:OLU:2756876.2629339} that colored predecessor search can be reduced to the rank query problem using partial sum data structures and clever mapping. In the rank query problem, a sequence $S[1,n]$ and an alphabet $[1,\sigma]$ is given with each $S[i] \in [1, \sigma]$ and after the preprocessing, given an index $i \in [1,n]$ and a symbol $c \in [1, \sigma]$ the data structure has to report the number of occurrences of $c$ in the subsequence $S[1,i)$. Finally, Gawrychowski et al. \cite{gawrychowski2016nearest} observed that the rank query problem can be reduced to the nearest colored node problem as follows: 
Given the sequence $S$, we build a path $P=(v_1, .., v_n)$ where $c(v_i)=S[i]$ for every $i \in [1,n]$. With each vertex $v \in P$, we store its rank. Consider a rank query of form $i \in [1,n], c \in [1, \sigma]$. Using a Color Distance Oracle, we query for the nearest colored vertex $v_j$ and if $j < i$, we return the rank stored at $v_j$. Otherwise, $j \geq i$ and we can return the rank stored with $v_j$ decreased by 1. Let us now take the proof for theorem \ref{thm:mainresultlbdistancesumry}.
\begin{proof}
We first consider the (approximate) nearest colored node problem on a weighted path and extend our proof later to the unweighted setting. We assume w.l.o.g.  that the approximation factor is $k \geq \log{n}$ and that $n$ is an exact power of $k$. We use the simple construction of the path $P=(v_1, .., v_n)$ from the sequence $S$ as before. We assume that the nearest colored vertex $v_j = \textsc{Nearest}(v_i, c)$ to $v_i$ has $j \geq i$. This is sufficient as we can reverse the path and run a second query giving the nearest colored vertex $v_{j'}$ with $j' \leq i$ and then take the closer one (in fact one procedure might not find a colored node but we always ensure that the nearest colored node is found in one of those two queries). We let $P[v_i, v_j]$ denote the subpath of $P$ between the vertices $v_i$ and $v_j$ and sometimes refer to it as interval $[i,j]$ of the path.

The intuition behind our proof is that we can choose the edge weights so that we can find the nearest colored node $v_j = \textsc{Nearest}(v_i, c)$ to $v_i$  even if we are only given a distance estimate $\widehat{\mathbf{dist}}(v_i, c)$ of approximation factor $k$. Therefore, we partition the path into intervals of size $k^l$ for every $0 \leq l \leq \log_k{n}$ and refer to them as level-$l$ intervals. We assign edge weights to delimit intervals. By construction, if an edge delimits two level-$l$ intervals then it also delimits two level-$l'$ intervals for all $l' \leq l$. Let an edge $(v_x, v_{x+1})$ be assigned weight $k^l$ if it delimits two level-$l$ intervals but doesn't delimit two level-$l+1$ intervals. More formally, an edge $(v_x, v_{x+1})$ is assigned the weight $k^l$ for the largest $k^l$ such that $k^l \mid x$. This is depicted in figure \ref{fig:path_construction}. Let us now observe that the heavy edges on the path dominate the path weight. A level-$l$ interval contains $k^{l-l'}$ level-$l'$ intervals for $l' < l$ and therefore there are $k^{l-l'} - k^{l-l'-1}$ edges of weight $k^{l'}$. Thus the path from the first node to the last node in the level-$l$ interval has weight
\[
\sum_{l' = 0}^{l-1} (k^{l-l'} - k^{l-l'-1}) k^{l'} < \sum_{l'=0}^{l-1} k^{l-l'} k^{'} = \sum_{l'=0}^{l-1} k^l = lk^l \leq \log_k{n}k^l \leq k^{l+1} 
\]
It follows that if we have $l = \lfloor \log_k(\mathbf{dist}(v_i,v_j)) \rfloor$ for $i \neq j$, then the path $P[v_i, v_j]$ contains an edge delimiting two level-$(l-1)$, i.e. of weight at least $k^{l-1}$, as otherwise the path costs would be strictly less than $k^l \leq \mathbf{dist}(v_i,v_j)$.

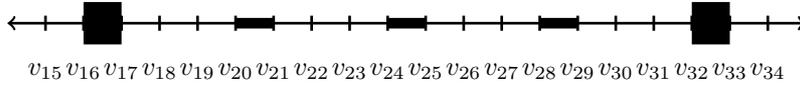
\begin{figure}
	\centering
        \begin{tikzpicture}[scale=0.5] 
        \draw[line width=1pt,<->](-2,0)--(19,0);
        \draw[line width=16pt](0,0)--(1,0); 
        \draw[line width=4pt](4,0)--(5,0); 
        \draw[line width=4pt](8,0)--(9,0); 
        \draw[line width=4pt](12,0)--(13,0); 
        \draw[line width=16pt](16,0)--(17,0); 
        
        \foreach \i in {15,...,34}{
            \draw[line width=1pt] (\i-16,-0.2)--(\i-16,0.2)node[pos=0,below=0.3cm]{$v_{\i}$};
        }
        \end{tikzpicture} 
	\caption[]{The weight of the edges on our constructed path $P$ where we chose $k = 4$. We illustrate heavy edge weight by increased boldness of the edge. For example the edge between vertices $v_{16}$ and $v_{17}$ has weight $k^2$ because $16 \mid k^2$ but $16 \nmid k^3$.}
  \label{fig:path_construction}
\end{figure}

Our second idea is that given two vertices $v_i, v_j$ with $i < j$ and path interval $[i,j]$ containing no node of color $c$, then the nearest colored vertex to $v_j$ is also the nearest colored vertex to $v_i$, i.e. $\textsc{Nearest}(v_i, c) = \textsc{Nearest}(v_j,c)$. We want to provide some special vertices that can return the nearest colored vertex but we need to do so carefully in order to retain near-linear space as our lower bound otherwise becomes meaningless. We therefore only cover intervals starting at special points. To simplify the presentation, we let $\iota(x, l)$ be the function that for any integer $x$ gives the next larger integer divisible by $k^l$. We store with each vertex $v_j$ where $j | k^l$ a data structure that can return $\textsc{Nearest}(v_j, c)$ for all colors $c$ that are on the path interval $[j, \iota(j,l+1)]$. We then say that $v_j$ covers $[j, \iota(j, l+1)]$ and observe that if $j$ is divisible by $k^l$, we cover all level-$l$ intervals starting at $v_j$ up to the end of the current level-$(l+1)$ interval. This is also depicted in figure \ref{fig:drawing_interval_cover}. In order to have fast look-ups, we store with each such vertex $v_j$ a hash map with an entry for each color $c \in C$ that occurs on the path and the corresponding vertex that is closest to $v_j$.

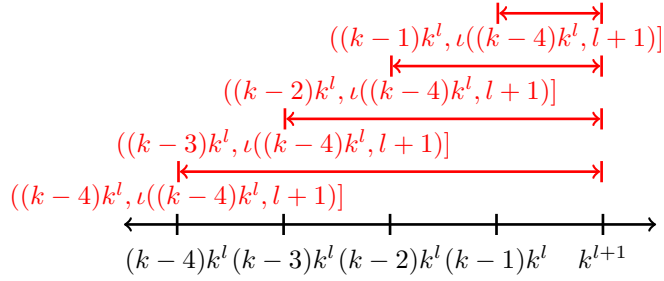
\begin{figure}
	\centering
        \begin{tikzpicture}[scale=0.7] 
          \draw[line width=1pt,<->](0,0)--(10,0); 
          \draw[line width=1pt] (1,-0.2)-- (1,0.2)node[pos=0,below]{$(k-4)k^l$};
          \draw[line width=1pt] (3,-0.2)-- (3,0.2)node[pos=0,below]{$(k-3)k^l$};
          \draw[line width=1pt] (5,-0.2)-- (5,0.2)node[pos=0,below]{$(k-2)k^l$};
          \draw[line width=1pt] (7,-0.2)-- (7,0.2)node[pos=0,below]{$(k-1)k^l$};
          \draw[line width=1pt] (9,-0.2)-- (9,0.2)node[pos=0,below]{$k^{l+1}$};

          \draw[line width=1pt, color=red, |<->|] (1,1)-- (9,1)node[pos=0,below]{$((k-4)k^l, \iota((k-4)k^l,l+1)]$};
          \draw[line width=1pt, color=red, |<->|] (3,2)-- (9,2)node[pos=0,below]{$((k-3)k^l, \iota((k-4)k^l,l+1)]$};
          \draw[line width=1pt, color=red, |<->|] (5,3)-- (9,3)node[pos=0,below]{$((k-2)k^l, \iota((k-4)k^l,l+1)]$};
          \draw[line width=1pt, color=red, |<->|] (7,4)-- (9,4)node[pos=0,below]{$((k-1)k^l, \iota((k-4)k^l,l+1)]$};
        \end{tikzpicture} 
	\caption[Interval coverage by using function $k$.]{The drawing shows a subinterval of $[1,n]$. We see that for every number $x$ that is divisible by $k^l$, we take the interval from $x$ to the closest number that is divisible by $k^{l+1}$ which is illustrated by the red interval. }
  \label{fig:drawing_interval_cover}
\end{figure}

It is straight-forward to see that given a vertex $v_j$, we can use the function $\iota$ to find $\log_k{n}$ special vertices $v_{\iota(j, l)}$ for $0 \leq l \leq \log_k{n}$ such that the union of all \textit{covered} intervals by those special vertices is $\bigcup_{0 \leq l \leq \log_k{n}} [\iota(j, l), \iota(j, l+1)] = [j, n]$ because $\iota(j, l)$ is divisible by $k^l$ by definition hence the hash map at vertex $v_{\iota(j,l)}$ covers the interval $[\iota(j,l), \iota(j, l+1)]$. Thus, we could already query the hash maps at these special vertices to extract the nearest colored node even without any distance estimate but it would incur $\log_k{n}$ look-ups.

Let us now combine both ideas to achieve that only a constant number of the associated hash maps at those special vertices need to be queried. Given a distance estimate $\widehat{\mathbf{dist}}(v_i,c)$ with $l = \lfloor\log_k{\widehat{\mathbf{dist}}(v_i,c)}\rfloor$, $v_j = \textsc{Nearest}(v_i, c)$. By our approximation guarantee 
\[\mathbf{dist}(v_i,v_j) \leq \widehat{\mathbf{dist}}(v_i,c) \leq k\mathbf{dist}(v_i,v_j)
\]
we get that $k^{l-1} \leq \mathbf{dist}(v_i,v_j) \leq k^{l+1}$. We conclude that it suffices to check the hash maps at the special vertices $v_{\iota(i, l')}$ from $l' \in \{l-1, l, l+1\}$ because $j \in [\iota(i, l-1), \iota(i, l+2)]$. Consider that this would not be the case, we know that the path from the interval $[i, \iota(i,l-1)]$ has path weight strictly less than $k^{l-1}$. If $j$ would be in interval $(\iota(i, l+2),n]$ then as the edge $(v_{\iota(i,l+2)}, v_{\iota(i,l+2)+1})$ has weight $k^{l+2}$ and every path to from $v_i$ to a vertex with index in that interval has to include this edge. In both cases, we derive a contradiction. 

It remains to prove that the space taken by the hash maps is near-linear. We therefore observe that the number of entries in each hash map is bounded by the size of the interval is has to cover as every node in the path interval has only one color. It is easy to see that we have $k^{\log_k{n}-j}$ vertices with indices divisible by $k^j$ and each covers an interval of size $k^{j+1}$.

Thus the number of total entries in all hash maps can be bounded by
\[
    \sum_{l = 0}^{\log_k{n}} k^{\log_k{n}-j} k^{j+1} = \sum_{l = 0}^{\log_k{n}} k^{\log_k{n}+1} = kn \log_k{n}
\]
As hash maps take space linear in the number of entries and $k = O(\textit{polylog}(n))$, we can bound the space by $\tilde{O}(n)$ incurring only a $\log\log{n}$ term in the lower bound as required. Thus, if the space is not dominated by the data structure for distance estimates, we still ensure the stated bounds. 

Finally, we observe that the path $P[v_1, v_n]$ has total weight at most 
\[
\log_k{n} k^{\log_k{n}} =n\log_k{n}
\]
as it is a level-$\log_k{n}$ interval. We can thus replace edges of weight $x$ with a path of $x-1$ dummy vertices and unit weight edges. 
\end{proof}

\section{Lower bounds for the dynamic setting}

During the last years, several techniques were presented to prove conditional lower bounds for dynamic problems by reducing to problems that are conjectured to be hard \cite{patrascu10mp-3sum,williams2010subcubic,DBLP:journals/corr/AbboudW14,DBLP:journals/corr/HenzingerKNS15}. We use the framework given by Henzinger et al. \cite{DBLP:journals/corr/HenzingerKNS15} who reduce their problems from a contrived version of Online-Vector-Multiplication defined as follows.
\begin{definition}
 [$\gamma$-OuMv problem (c.f. Definition 2.6 \cite{DBLP:journals/corr/HenzingerKNS15})] Let $\gamma > 0$ be a fixed constant. An algorithm for the $\gamma$-OuMv problem is given parameters $n_1,n_2, n_3$ as its input with the promise that $n_1 =\lfloor n_2^\gamma \rfloor$. Next, it is given a matrix $M$ of size $n_1 \times n_2$ that can be preprocessed. Let $p(n_1 , n_2)$ denote the preprocessing time. After the preprocessing, a sequence of vector pairs $(u^1,v^1), .., (u^{n_3}, v^{n_3})$ is presented one vector pair after another and the task is to compute $(u^t)^{\intercal} Mv^t$, before the pair $(u^{t+1}, v^{t+1})$ arrives. Let $c(n_1 , n_2, n_3 )$ denote the computation time over the whole sequence. The special case where $n_3=1$ is called the $\gamma$-uMv problem.
\end{definition}
They then show that any algorithm solving the $\gamma$-OuMv problem in $\tilde{\tilde{o}}(n_1 n_2 n_3)$ time would give a truly subcubic algorithm to solve Online Vector Multiplication(c.f. Theorem 2.2  \cite{DBLP:journals/corr/HenzingerKNS15}). 

As stated in the introduction, we consider in this section the existential version of nearest colored node that is we only ask whether there exists a colored vertex in the same component as a vertex $v$. As shown by Abboud and Williams \cite{DBLP:journals/corr/AbboudW14} it suffices to prove lower bounds on the worst-case update and query time for a partially-dynamic version of a problem to establish amortized lower bounds for the fully-dynamic version. We thus only prove hardness of the partially-dynamic settings where we reduce from the $\gamma$-uMv problem where we are given a $n_1 \times n_2$ matrix $\mathbf{M}$ to preprocess and only a single pair of vectors $(\mathbf{u}, \mathbf{v})$ arrives.

\begin{lemma}
\label{lemma:partdynalgoquery}
Given any algorithm $\mathcal{A}$ that is able to process updates in $\textsc{u}(n,\sigma) =\tilde{\tilde{o}}(\sigma)$ and queries in $\textsc{q}(n,\sigma)= \tilde{\tilde{o}}(n/\sigma)$ amortized time, we can solve $\gamma$-OuMv in time $\tilde{\tilde{o}}(n_1, n_2, n_3)$. The same lower bounds extend to the worst-case update and query times of the partially-dynamic version of the problem. 
\end{lemma}
\begin{proof}
We first focus on the decremental setting and extend the proof for the fully-dynamic and incremental version. Recall that we are given a $n_1 \times n_2$ matrix $\mathbf{M}$. We treat the $i$th row of $\mathbf{M}$ as a subset of $[1,n_2]$, i.e. $\mathbf{M}[i,j] = 1$ iff $j \in \mathbf{M}[i]$. We create a graph $G$ with a coloring function $c : V \rightarrow C$, with color set $C = \{c_1, .., c_{n_2}\}$, as follows: We first create a special vertex $s$. For every row $i$, we create for each $j \in \mathbf{M}[i]$ a vertex with c $c_j \in C$ and connect the vertices by linking every two consecutive vertices created. The created component for the $i$th row forms a simple path and is denoted from hereon by $P_i$. We also include an edge from $s$ to the first vertex on the path $P_i$ for every $i$. This completes the preprocessing phase. Clearly, the graph $G$ has at most $O(n_1n_2)$ vertices and as the construction forms a tree, we also have at most $O(n_1n_2)$ edges. The complete set-up is depicted in figure \ref{fig:drawingcondlb}.
\begin{figure}
	\centering
        \begin{tikzpicture}[scale=0.8,darkstyle/.style={circle,draw,fill=gray!40,minimum size=20}]
          \node [darkstyle]  (a1) at (1.5*1,-1.5*1){  $c_1$};
          \node [darkstyle]  (a2) at (1.5*2,-1.5*1){  $c_2$};
          \node [darkstyle]  (a3) at (1.5*4,-1.5*1){  $c_4$};
          \node [darkstyle]  (a4) at (1.5*6,-1.5*1){  $c_6$};

          \draw  (a1)--(a2) ;
          \draw  (a2)--(a3) ;
          \draw  (a3)--(a4) ;

          \node [darkstyle]  (b1) at (1.5*1,-1.5*2){  $c_1$};
          \node [darkstyle]  (b2) at (1.5*3,-1.5*2){  $c_3$};
          \node [darkstyle]  (b3) at (1.5*4,-1.5*2){  $c_4$};
          \node [darkstyle]  (b4) at (1.5*5,-1.5*2){  $c_5$};

          \draw  (b1)--(b2) ;
          \draw  (b2)--(b3) ;
          \draw  (b3)--(b4) ;

          \node [darkstyle]  (c1) at (1.5*3,-1.5*3){  $c_3$};
          \node [darkstyle]  (c2) at (1.5*5,-1.5*3){  $c_5$};

          \draw  (c1)--(c2) ;

          \node [darkstyle]  (d1) at (1.5*1,-1.5*4){  $c_1$};

          \node [darkstyle]  (e1) at (1.5*2,-1.5*5){  $c_2$};
          \node [darkstyle]  (e2) at (1.5*3,-1.5*5){  $c_3$};
          \node [darkstyle]  (e3) at (1.5*4,-1.5*5){  $c_4$};
          \node [darkstyle]  (e4) at (1.5*5,-1.5*5){  $c_5$};
          \node [darkstyle]  (e5) at (1.5*6,-1.5*5){  $c_6$};

          \draw  (e1)--(e2) ;
          \draw  (e2)--(e3) ;
          \draw  (e3)--(e4) ;
          \draw  (e4)--(e5) ;
     
        \node [darkstyle]  (s) at (0,-3){$s$};
     
          \draw[very thick, red, dashed]  (s)--(1, -1.5 * 1)--(a1) ;
          \draw[very thick, red, dashed]  (s)--(1, -1.5 * 2)--(b1) ;
          \draw[very thick, red, dashed]  (s)--(1, -1.5 * 3)--(c1) ;
          \draw[very thick, red, dashed]  (s)--(1, -1.5 * 4)--(d1) ;
          \draw[very thick, red, dashed]  (s)--(1, -1.5 * 5)--(e1) ;

        \end{tikzpicture} 
	\caption[Graph for emulating Matrix Multiplication via the decremental vertex-label connectivity problem.]{Depiction of the set-up of a graph from a $5 \times 6$ matrix. The point $s$ is used as query point and in the decremental case all dashed red edges are initially in the graph and can be deleted depending on $\mathbf{u}$.}
  \label{fig:drawingcondlb}
\end{figure}
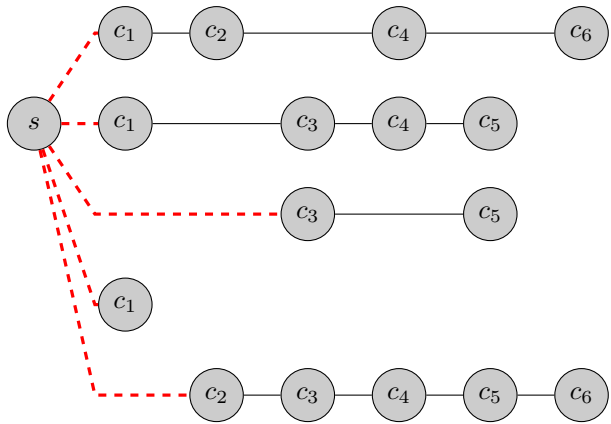

Consider that a vector pair $(\mathbf{u},\mathbf{v})$ arrives. For each $i$ where $\mathbf{u}[i]=0$, we remove the edge from the first vertex in $R_i$ to vertex $s$. This incurs at most $n_1$ updates. We have $\mathbf{u}^\intercal \mathbf{Mv} = 1$ if and only if there exists a $j$ with $\mathbf{v}[j]=1$ where $s$ is still connected to the color $c_j$. To check these connections, we need at most $n_2$ queries. 

As shown by Abboud and Williams \cite{DBLP:journals/corr/AbboudW14}, we can run our algorithm on a machine that records all changes and reverts them after a single pair $(\mathbf{u}, \mathbf{v})$ is processed to the original state in time proportional to the running time since the original state was left. Our graph has at most $\sigma = n_2$ different colors. It is now straight-forward to see that we can solve $\gamma$-OuMv in time $O(n_3 (n_1 \textsc{u}(n_1\sigma, \sigma) + \sigma \textsc{q}(n_1\sigma, \sigma))$ thus the stated bound follows.

In the incremental setting, we omit the edges adjacent to $s$ initially. Then, when $(\mathbf{u},\mathbf{v})$ arrives, let $\mathbf{u}$ have the indices $u_1, u_2, ..$ set to 1. Then, we join the first vertices in $R_i$ and $R_{i+1}$ for even $i$ and the last vertices of $R_i$ and $R_{i+1}$ for odd $i$. Thus, we only construct a path. Connecting $s$ to the end of the path, allows us to run queries as in the decremental setting. For the fully-dynamic setting, we can use the same set-up as in the incremental setting but instead of rolling the edge insertions back in each phase, we can simply run edge deletions to recover the original state implying that we get an amortized bound.
\end{proof}
We underline the generality of our lower bound which establishes that even on path graphs the amortized fully-dynamic problem remains hard. To strengthen our lower bound, we observe that we can decrease the graph size to $O(n_1)$ for directed graphs or if we are given distance estimates of approximation $<5/3$ implying that the query bound in the theorem can be replaced by $\tilde{\tilde{o}}(n)$. We therefore create $\sigma = n_2$ vertices $V_c$, one of each color. Instead of constructing an entire row $R_i$ for the $i$'th set, we construct a single vertex $v_i$ and connect it with edges to vertices in $V_c$ that match elements in $v_i$. We can then run the algorithm as before. If we direct the edge from $s$ to each $v_i$ and from each $v_i$ towards the vertices in $V_c$ our algorithm works as before. For undirected graphs, we have that $\mathbf{dist}(v_i, c_j) = 3$ iff $\mathbf{u}^\intercal \mathbf{Mv} = 1$ and otherwise $\mathbf{dist}(v_i, c_j) \geq 5$. Thus any approximation of factor smaller $5/3$ is still sufficient to distinguish the two cases. Clearly the graph has $O(n_1+n_2) = O(n_1)$ vertices as $\sigma \leq n$. 

We point out that the underlying $OMv$-conjecture even applies in case of error probability $1/3$ thus our lower bound applies even to Monte-Carlo algorithms. Interestingly, the directed incremental version of our problem can be seen as a graph version of Pagh's problem (we follow the definition from \cite{patrascu10mp-3sum}). In Pagh's problem, we are given a collection $\mathcal{C}$ of $k$ sets $C_1, C_2,.., C_k \subseteq [n]$. We are then allowed to update by providing two indices $i,j \in \{1,..,k\}$ adding the set $C_i \cap C_j$ to $\mathcal{C}$. We then want to be able to query given $x \in [n]$, $i \in \{1,..,k\}$ if $x \in C_i$. Similarly to the proof let us assume that each set $C_i$ is represented by a path $P_i$ containing a vertex of color $c$ for each $c \in \overline{C_i}$. Then updates for $i,j \in \{1,..,k\}$ can be implemented by adding a new vertex and an edge from it to the beginning of $P_i$ and one to the first vertex in $P_j$. Queries can be implemented by asking whether a node of color $x \in [n]$ can be reached from the first node of $P_i$ and by returning the negated answer.

\subparagraph*{Acknowledgements.}

I want to thank Christian Wulff-Nilsen for inspiring the research on this problem and the guidance during the project.

\appendix

%%
%% Bibliography
%%

%% Either use bibtex (recommended), 

\bibliography{nearestColoredNode}

\begin{thebibliography}{10}

\bibitem{DBLP:journals/corr/AbboudW14}
Amir Abboud and Virginia~Vassilevska Williams.
\newblock Popular conjectures imply strong lower bounds for dynamic problems.
\newblock In {\em Foundations of Computer Science (FOCS), 2014 IEEE 55th Annual
  Symposium on}, pages 434--443. IEEE, 2014.

\bibitem{DBLP:journals/corr/AbrahamCEFN17}
Ittai Abraham, Shiri Chechik, Michael Elkin, Arnold Filtser, and Ofer Neiman.
\newblock Ramsey spanning trees and their applications.
\newblock In {\em Proceedings of the Twenty-Ninth Annual ACM-SIAM Symposium on
  Discrete Algorithms}, pages 1650--1664. SIAM, 2018.

\bibitem{Bartal:2003:MRP:780542.780610}
Yair Bartal, Nathan Linial, Manor Mendel, and Assaf Naor.
\newblock On metric ramsey-type phenomena.
\newblock In {\em Proceedings of the Thirty-fifth Annual ACM Symposium on
  Theory of Computing}, STOC '03, pages 463--472, New York, NY, USA, 2003. ACM.
\newblock URL: \url{http://doi.acm.org/10.1145/780542.780610}, \href
  {http://dx.doi.org/10.1145/780542.780610} {\path{doi:10.1145/780542.780610}}.

\bibitem{Belazzougui:2015:OLU:2756876.2629339}
Djamal Belazzougui and Gonzalo Navarro.
\newblock Optimal lower and upper bounds for representing sequences.
\newblock {\em ACM Trans. Algorithms}, 11(4):31:1--31:21, April 2015.
\newblock URL: \url{http://doi.acm.org/10.1145/2629339}, \href
  {http://dx.doi.org/10.1145/2629339} {\path{doi:10.1145/2629339}}.

\bibitem{Bender:2000:LPR:646388.690192}
Michael~A. Bender and Martin Farach-Colton.
\newblock The lca problem revisited.
\newblock In {\em Proceedings of the 4th Latin American Symposium on
  Theoretical Informatics}, LATIN '00, pages 88--94, London, UK, UK, 2000.
  Springer-Verlag.
\newblock URL: \url{http://dl.acm.org/citation.cfm?id=646388.690192}.

\bibitem{DBLP:journals/corr/abs-1109-3114}
Shiri Chechik.
\newblock Improved distance oracles and spanners for vertex-labeled graphs.
\newblock In {\em Proceedings of the 20th Annual European Conference on
  Algorithms}, ESA'12, pages 325--336, Berlin, Heidelberg, 2012.
  Springer-Verlag.
\newblock URL: \url{http://dx.doi.org/10.1007/978-3-642-33090-2\_29}, \href
  {http://dx.doi.org/10.1007/978-3-642-33090-2\_29}
  {\path{doi:10.1007/978-3-642-33090-2\_29}}.

\bibitem{DBLP:journals/corr/abs-1305-3314}
Shiri Chechik.
\newblock Approximate distance oracles with constant query time.
\newblock In {\em Proceedings of the forty-sixth annual ACM symposium on Theory
  of computing}, pages 654--663. ACM, 2014.

\bibitem{chechik2015approximate}
Shiri Chechik.
\newblock Approximate distance oracles with improved bounds.
\newblock In {\em Proceedings of the Forty-Seventh Annual ACM on Symposium on
  Theory of Computing}, pages 1--10. ACM, 2015.

\bibitem{dietzfelbinger1994dynamic}
Martin Dietzfelbinger, Anna Karlin, Kurt Mehlhorn, Friedhelm Meyer auF~der
  Heide, Hans Rohnert, and Robert~E Tarjan.
\newblock Dynamic perfect hashing: Upper and lower bounds.
\newblock {\em SIAM Journal on Computing}, 23(4):738--761, 1994.

\bibitem{Fischer2006}
Johannes Fischer and Volker Heun.
\newblock {\em Theoretical and Practical Improvements on the RMQ-Problem, with
  Applications to LCA and LCE}, pages 36--48.
\newblock Springer Berlin Heidelberg, Berlin, Heidelberg, 2006.
\newblock URL: \url{http://dx.doi.org/10.1007/11780441\_5}, \href
  {http://dx.doi.org/10.1007/11780441\_5} {\path{doi:10.1007/11780441\_5}}.

\bibitem{gawrychowski2016nearest}
Pawel Gawrychowski, Gad~M Landau, Shay Mozes, and Oren Weimann.
\newblock The nearest colored node in a tree.
\newblock In {\em LIPIcs-Leibniz International Proceedings in Informatics},
  volume~54. Schloss Dagstuhl-Leibniz-Zentrum fuer Informatik, 2016.

\bibitem{DBLP:journals/corr/HenzingerKNS15}
Monika Henzinger, Sebastian Krinninger, Danupon Nanongkai, and Thatchaphol
  Saranurak.
\newblock Unifying and strengthening hardness for dynamic problems via the
  online matrix-vector multiplication conjecture.
\newblock In {\em Proceedings of the forty-seventh annual ACM symposium on
  Theory of computing}, pages 21--30. ACM, 2015.

\bibitem{Hermelin2011}
Danny Hermelin, Avivit Levy, Oren Weimann, and Raphael Yuster.
\newblock {\em Distance Oracles for Vertex-Labeled Graphs}, pages 490--501.
\newblock Springer Berlin Heidelberg, Berlin, Heidelberg, 2011.
\newblock URL: \url{http://dx.doi.org/10.1007/978-3-642-22012-8\_39}, \href
  {http://dx.doi.org/10.1007/978-3-642-22012-8\_39}
  {\path{doi:10.1007/978-3-642-22012-8\_39}}.

\bibitem{laish2017efficient}
Itay Laish and Shay Mozes.
\newblock Efficient approximate distance oracles for vertex-labeled planar
  graphs.
\newblock {\em arXiv preprint arXiv:1707.02414}, 2017.

\bibitem{DBLP:journals/corr/abs-cs-0511084}
Manor Mendel and Assaf Naor.
\newblock Ramsey partitions and proximity data structures.
\newblock In {\em Foundations of Computer Science, 2006. FOCS'06. 47th Annual
  IEEE Symposium on}, pages 109--118. IEEE, 2006.

\bibitem{DBLP:journals/corr/MozesS15}
Shay Mozes and Eyal~E Skop.
\newblock Efficient vertex-label distance oracles for planar graphs.
\newblock {\em Theory of Computing Systems}, 62(2):419--440, 2018.

\bibitem{patrascu10mp-3sum}
Mihai P{\v a}tra{\c s}cu.
\newblock Towards polynomial lower bounds for dynamic problems.
\newblock In {\em Proc. 42nd ACM Symposium on Theory of Computing (STOC)},
  pages 603--610, 2010.

\bibitem{DBLP:journals/corr/abs-cs-0603043}
Mihai P{\u{a}}tra{\c{s}}cu and Mikkel Thorup.
\newblock Time-space trade-offs for predecessor search.
\newblock In {\em Proceedings of the thirty-eighth annual ACM symposium on
  Theory of computing}, pages 232--240. ACM, 2006.

\bibitem{patracscu2007randomization}
Mihai Pǎtra{\c{s}}cu and Mikkel Thorup.
\newblock Randomization does not help searching predecessors.
\newblock In {\em Proceedings of the eighteenth annual ACM-SIAM symposium on
  Discrete algorithms}, pages 555--564. Society for Industrial and Applied
  Mathematics, 2007.

\bibitem{roditty2005deterministic}
Liam Roditty, Mikkel Thorup, and Uri Zwick.
\newblock Deterministic constructions of approximate distance oracles and
  spanners.
\newblock In {\em International Colloquium on Automata, Languages, and
  Programming}, pages 261--272. Springer, 2005.

\bibitem{Thorup:2001:CRS:378580.378581}
Mikkel Thorup and Uri Zwick.
\newblock Compact routing schemes.
\newblock In {\em Proceedings of the Thirteenth Annual ACM Symposium on
  Parallel Algorithms and Architectures}, SPAA '01, pages 1--10, New York, NY,
  USA, 2001. ACM.
\newblock URL: \url{http://doi.acm.org/10.1145/378580.378581}, \href
  {http://dx.doi.org/10.1145/378580.378581} {\path{doi:10.1145/378580.378581}}.

\bibitem{Thorup:2005:ADO:1044731.1044732}
Mikkel Thorup and Uri Zwick.
\newblock Approximate distance oracles.
\newblock {\em J. ACM}, 52(1):1--24, January 2005.
\newblock URL: \url{http://doi.acm.org/10.1145/1044731.1044732}, \href
  {http://dx.doi.org/10.1145/1044731.1044732}
  {\path{doi:10.1145/1044731.1044732}}.

\bibitem{tsur2016succinct}
Dekel Tsur.
\newblock Succinct data structures for nearest colored node in a tree.
\newblock {\em Information Processing Letters}, 132:6--10, 2018.

\bibitem{williams2010subcubic}
Virginia~Vassilevska Williams and Ryan Williams.
\newblock Subcubic equivalences between path, matrix and triangle problems.
\newblock In {\em Foundations of Computer Science (FOCS), 2010 51st Annual IEEE
  Symposium on}, pages 645--654. IEEE, 2010.

\bibitem{DBLP:journals/corr/abs-1202-2336}
Christian Wulff-Nilsen.
\newblock Approximate distance oracles with improved query time.
\newblock In {\em Proceedings of the twenty-fourth annual ACM-SIAM symposium on
  Discrete algorithms}, pages 539--549. SIAM, 2013.

\end{thebibliography}

%% .. or use the thebibliography environment explicitely

\section{Proof of Lemma \ref{lem:sizebunch}}

\begin{lemma}
 We use at most space $O(kn\sigma^{1/k})$ to represent the Color Distance Oracle and construction time $O(m\sigma)$.
\end{lemma}
\begin{proof}
We first want to prove that $\sum_{c \in C} |B(c)| = O(kn\sigma^{1/k})$. We follow the analysis by Thorup and Zwick to bound the bunch sizes for a vertex $v \in V$ but omit the last layer that we account for in a later step. They show that the expected size $E[|B(v) \cap (A_i \setminus A_{i+1})|] = \sigma^{1/k}$, for $1 \leq i \leq k - 2$ and that $E[|A_{k-1}|] = n\sigma^{(k-1)/k}$. By linearity of expectation, we derive $E[|\bigcup_{i < k - 1} B(v) \cap  (A_i \setminus A_{i+1})|] = E[|B(v) \setminus A_{k-1}|] = (k-1)\sigma^{1/k}$. Now, observe that $A_{k-1} \setminus A_k = A_{k-1}$, therefore $A_{k-1} \subseteq B(v)$ for every bunch $B(v)$. But we only need to account for the set $A_{k-1}$ once per bunch $B(c)$. We obtain 
  \begin{align*}
  \sum_{c \in C} E[|B(c)|] &\leq  \sum_{c \in C}\left(\sum_{v \in V_c} E[|B(v) \setminus A_{k-1}|] + E[|A_{k-1}|]\right) \\ &=  \sum_{v \in V} E[|B(v) \setminus A_{k-1}|] +  \sum_{c \in C} E[|A_{k-1}|] \\&=  (k-1) n\sigma^{1/k} +  n\sigma^{1/k} = kn\sigma^{1/k} 
  \end{align*}
As a hash map only takes linear space in the number of entries, the hash tables take space proportional to the number of elements in the bunches. We also have to account for storing the RMQ structures but as they take linear space in the size of the arrays this uses at most $O(kn)$ space overall which is subsumed in $O(kn\sigma^{1/k})$. All distances can be precomputed by computing a look-up table in $\tilde{O}(m\sigma)$ time which dominates the construction costs for large $\sigma$.
\end{proof}

\section{Proof of Lemma \ref{lmm:treecontainlambda}}

\begin{lemma}
We can maintain a data structure over a tree $T=(V,E)$ with function $c : V \rightarrow C$ as defined before, that given $v \in V, c \in C$ finds the nearest ancestor of $v$ in $T$ that has a $c$-colored vertex in its subtree and that is able to process color-reassignments. Both operations take $O(\log\log{n})$ worst-case time and the data structure requires $O(n)$ space.
\end{lemma}
\begin{proof}
We initialize the data structure by augmenting the tree $T$ with a LCA data structure. We then number the vertices in the order they occur in a BFS and denote the assigned number by $\lambda(v)$ for $v \in V$. Finally, we construct for each color a predecessor data structure and insert the numbers of all vertices of color $c \in C$ for the corresponding data structure. As a predecessor structure takes only space linear in the number of elements the size of all predecessor structures sums to $O(n)$. Given a query from a vertex $v \in V$ for color $c$, we use the predecessor structure for color $c$ and search for the predecessor $p$ and successor $s$ of $\lambda(v)$. It is straight-forward to see that the ancestor $a \in \{\textsc{Lca}(v,p), \textsc{Lca}(v,s)\}$ at lower depth gives the desired ancestor. For updating the color $c$ to $c'$ for vertex $v$, we remove $\lambda(v)$ from the predecessor structure for color $c$ and insert it into the structure for color $c'$. 
\end{proof}

\end{document}